\documentclass[conference]{IEEEtran}

\usepackage{cite}
\usepackage{graphicx}
\usepackage{xcolor}
\usepackage{tikz}
\usetikzlibrary{arrows,positioning,shapes.geometric}
\usepackage{pgfplots}
\usepackage{multirow}
\usepackage{upgreek}
\usepackage{amssymb}
\usepackage{amsmath}
\usepackage{cases}
\usepackage{pifont}
\usepackage{bm}
\usepackage{amsthm}
\newtheorem{theorem}{Theorem}
\newtheorem{lemma}{Lemma}
\newtheorem{remark}{Remark}
\usepackage{tabularx}
\usepackage{booktabs}
\newcolumntype{Y}{>{\centering\arraybackslash}X}
\usepackage{algorithmic}
\usepackage{array}

\newcommand{\fixme}[2]{\ifx&#2&{\leavevmode\color{red}#1}\else{\leavevmode\color{red}FIXME\{}#1{\leavevmode\color{red}\}}\footnote{{\leavevmode\color{red}#2}}\PackageWarning{Fixme}{#1: #2}\fi}

\DeclareMathOperator*{\sgn}{sgn}
\DeclareMathOperator{\PM}{PM}
\DeclareMathOperator*{\arctanh}{arctanh}

\begin{document}

\title{Fast Simplified Successive-Cancellation List Decoding of Polar Codes}

\author{\IEEEauthorblockN{Seyyed Ali Hashemi, Carlo Condo, Warren J. Gross}
\IEEEauthorblockA{Department of Electrical and Computer Engineering, McGill University, Montr\'eal, Qu\'ebec, Canada\\
Email: seyyed.hashemi@mail.mcgill.ca, carlo.condo@mail.mcgill.ca, warren.gross@mcgill.ca}}

\maketitle

\begin{abstract}

Polar codes are capacity achieving error correcting codes that can be decoded through the successive-cancellation algorithm. To improve its error-correction performance, a list-based version called successive-cancellation list (SCL) has been proposed in the past, that however substantially increases the number of time-steps in the decoding process. The simplified SCL (SSCL) decoding algorithm exploits constituent codes within the polar code structure to greatly reduce the required number of time-steps without introducing any error-correction performance loss.
In this paper, we propose a faster decoding approach to decode one of these constituent codes, the \mbox{Rate-1} node. We use this \mbox{Rate-1} node decoder to develop \mbox{Fast-SSCL}. We demonstrate that only a list-size-bound number of bits needs to be estimated in \mbox{Rate-1} nodes and \mbox{Fast-SSCL} exactly matches the error-correction performance of SCL and SSCL. This technique can potentially greatly reduce the total number of time-steps needed for polar codes decoding: analysis on a set of case studies show that \mbox{Fast-SSCL} has a number of time-steps requirement that is up to $\bm{66.6\%}$ lower than SSCL and $\bm{88.1\%}$ lower than SCL.

\end{abstract}

\IEEEpeerreviewmaketitle

\section{Introduction} \label{sec:intro}

Polar codes are a class of error-correction codes introduced by Ar{\i}kan in \cite{arikan}. They can provably achieve channel capacity on a memoryless channel when the code length $N$ tends to infinity. The first polar code decoding algorithm to be proposed is the successive-cancellation (SC), that can be represented as a binary tree search with complexity $O(N \log_2 N)$. The full search can be completed in $2N-2$ time-steps \cite{leroux}. Various works in the past \cite{alamdar, sarkis} have analyzed the nature of the nodes of the SC tree, noting that nodes whose leaves present certain patterns of information and redundancy bits, do not need to be traversed. 

While SC decoding is very effective when applied to polar codes with $N\rightarrow \infty$, its error-correction performance degrades very quickly with short and medium codes. Alternative decoding algorithms have been proposed to overcome this issue, among which successive-cancellation list (SCL) is one of the most promising \cite{tal_list}: instead of focusing on a single candidate codeword like SC, the $L$ most probable candidate codewords are allowed to survive concurrently. The error-correction performance of polar codes under SCL decoding, when concatenated with a cyclic redundancy check, has been shown to be comparable to that of some low-density parity-check codes used in current standards. SCL yields better error-correction performance than SC at the cost of additional latency, requiring $2N+K-2$ time-steps to be completed \cite{balatsoukas}, where $K$ is the number of information bits in the code. The technique proposed in \cite{sarkis_list} applies the tree pruning methods devised for SC to SCL, but is based on heuristics and needs to be redesigned every time code parameters are modified.

The authors proposed in \cite{hashemi_LSD, hashemi_MR} a sphere-based approach to list decoding of polar codes, that has led to the development and implementation of the simplified successive-cancellation list (SSCL) decoding algorithm \cite{hashemi_SSCL,hashemi_SSCL_TCASI}. SSCL guarantees significant reduction in the number of required time-steps with respect to SCL without relying on approximations or code-specific design. Thus, it can be applied to any code and yields exactly the same error-correction performance of SCL.

This work proposes a simpler decoder for one of the special nodes used in SCL and SSCL, the \mbox{Rate-1} node. Without any kind of error-correction performance degradation, it is able to decode a \mbox{Rate-1} node of length $N_\nu$ in $\min \left(L-1,N_\nu\right)$ time-steps, against the $N_\nu$ and $3N_\nu-2$ required by SSCL and SCL respectively. This \mbox{Rate-1} node decoder is then used instead of the \mbox{Rate-1} node decoder of SSCL. We call the decoder that incorporates the new \mbox{Rate-1} decoder ``\mbox{Fast-SSCL}". Given that in practical polar codes there are many instances of \mbox{Rate-1} nodes where $L \ll N_\nu$, we show that the proposed \mbox{Fast-SSCL} can speed up the SSCL decoder of up to $66.6\%$.

The rest of the paper is organized as follows. Section~\ref{sec:polar} briefly introduces polar codes and the SC, SCL and SSCL decoding algorithms. Section~\ref{sec:FSSCL} describes the novel decoding approach for the \mbox{Rate-1} node, and provides proof of its exactness. In Section~\ref{sec:impr} the reduction in the number of decoding time-steps is quantified and compared to previous results for a set of polar codes. Section~\ref{sec:conc} draws the conclusions.

\section{Polar Codes Encoding and Decoding}\label{sec:polar}

A polar code is represented by $\mathcal{P}(N,K)$ and can be constructed by concatenating two polar codes of length $N/2$. This recursive construction can be denoted as a matrix multiplication as $\mathbf{x} = \mathbf{u}\mathbf{G}^{\otimes n}$, where $\mathbf{u} = \{u_0,u_1,\ldots,u_{N-1}\}$ is the sequence of input bits, $\mathbf{x} = \{x_0,x_1,\ldots,x_{N-1}\}$ is the sequence of coded bits, and $\mathbf{G}^{\otimes n}$ is the $n$-th Kronecker product of the polarizing matrix $\mathbf{G}=\bigl[\begin{smallmatrix} 1&0\\ 1&1 \end{smallmatrix} \bigr]$. The encoding process involves the determination of the $K$ bit-channels with the best channel characteristics and assigning the information bits to them. The remaining $N-K$ bit-channels are set to a known value known at the decoder side. They are thus called frozen bits with set $\mathcal{F}$. Since the value of these bits does not have an impact on the error-correction performance of polar codes on a symmetric channel, they are usually set to $0$. The codeword $\mathbf{x}$ is then modulated and sent through the channel. In this paper, we consider BPSK modulation which maps $\{0,1\}$ to $\{+1,-1\}$.

\subsection{Successive-Cancellation Decoding}
\label{sec:polar:SCDec}

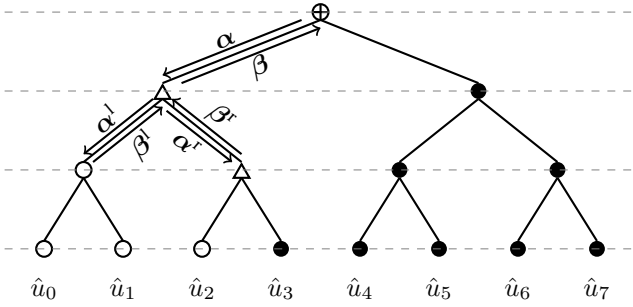
\begin{figure}
  \centering
  \begin{tikzpicture}[scale=2.1, thick]
\newcommand\Triangle[1]{-- ++(0:2*#1) -- ++(120:2*#1) --cycle}

  \draw (0,0) circle [radius=.05];
  
  \draw (-.05,0) -- (.05,0);
  \draw (0,-.05) -- (0,.05);

  \draw (-1.05,-.55) \Triangle{.05};
  \fill (1,-.5) circle [radius=.05];

  \draw (-1.5,-1) circle [radius=.05];
  \draw (-.55,-1.05) \Triangle{.05};
  \fill (.5,-1) circle [radius=.05];
  \fill (1.5,-1) circle [radius=.05];

  \draw (-1.75,-1.5) circle [radius=.05];
  \draw (-1.25,-1.5) circle [radius=.05];
  \draw (-.75,-1.5) circle [radius=.05];
  \fill (-.25,-1.5) circle [radius=.05];
  \fill (.25,-1.5) circle [radius=.05];
  \fill (.75,-1.5) circle [radius=.05];
  \fill (1.25,-1.5) circle [radius=.05];
  \fill (1.75,-1.5) circle [radius=.05];

  \node at (-1.75,-1.75) {$\hat{u}_0$};
  \node at (-1.25,-1.75) {$\hat{u}_1$};
  \node at (-.75,-1.75) {$\hat{u}_2$};
  \node at (-.25,-1.75) {$\hat{u}_3$};
  \node at (.25,-1.75) {$\hat{u}_4$};
  \node at (.75,-1.75) {$\hat{u}_5$};
  \node at (1.25,-1.75) {$\hat{u}_6$};
  \node at (1.75,-1.75) {$\hat{u}_7$};

  \draw (0,-.05) -- (-1,-.45);
  \draw (0,-.05) -- (1,-.45);

  \draw (-1,-.55) -- (-1.5,-.95);
  \draw (-1,-.55) -- (-.5,-.95);
  \draw (1,-.55) -- (.5,-.95);
  \draw (1,-.55) -- (1.5,-.95);

  \draw (-1.5,-1.05) -- (-1.75,-1.45);
  \draw (-1.5,-1.05) -- (-1.25,-1.45);
  \draw (-.5,-1.05) -- (-.75,-1.45);
  \draw (-.5,-1.05) -- (-.25,-1.45);
  \draw (.5,-1.05) -- (.25,-1.45);
  \draw (.5,-1.05) -- (.75,-1.45);
  \draw (1.5,-1.05) -- (1.25,-1.45);
  \draw (1.5,-1.05) -- (1.75,-1.45);

  \draw [very thin,gray,dashed] (-2,0) -- (2,0);
  \draw [very thin,gray,dashed] (-2,-.5) -- (2,-.5);
  \draw [very thin,gray,dashed] (-2,-1) -- (2,-1);
  \draw [very thin,gray,dashed] (-2,-1.5) -- (2,-1.5);

  \draw [->] (-.12,-.05) -- (-1,-.4) node [above=-.1cm,midway,rotate=25] {$\bm{\alpha}$};
  \draw [->] (-.88,-.45) -- (0,-.1) node [below=-.1cm,midway,rotate=25] {$\bm{\beta}$};

  \draw [->] (-1.06,-.55) -- (-1.5,-.9) node [above=-.1cm,midway,rotate=40] {$\bm{\alpha}^{\text{l}}$};
  \draw [->] (-1.44,-.95) -- (-1.0,-0.6) node [below=-.1cm,midway,rotate=40] {$\bm{\beta}^{\text{l}}$};

  \draw [<-] (-.94,-.55) -- (-.5,-.9) node [above=-.1cm,midway,rotate=-40] {$\bm{\beta}^{\text{r}}$};
  \draw [<-] (-.56,-.95) -- (-0.975,-.625) node [below=-.1cm,midway,rotate=-40] {$\bm{\alpha}^{\text{r}}$};

\end{tikzpicture}
  \caption{SC decoding on a binary tree for $\mathcal{P}(8,5)$ and $\{u_0,u_1,u_2\}\in \mathcal{F}$.}
  \label{figSCDec}
\end{figure}

The binary tree shown in Fig.~\ref{figSCDec} represents the SC decoding process of $\mathcal{P}(8,5)$. For a node of length $N_\nu$, soft logarithmic likelihood ratio (LLR) values $\bm{\alpha} = \{\alpha_0,\alpha_1,\ldots,\alpha_{N_\nu-1}\}$ pass from parent to child nodes, while the hard bit estimates $\bm{\beta} = \{\beta_0,\beta_1,\ldots,\beta_{N_\nu-1}\}$ follow the opposite direction.

The $\frac{N_\nu}{2}$ elements vectors $\bm{\alpha}^\text{l} = \{\alpha^\text{l}_0,\alpha^\text{l}_1,\ldots,\alpha^\text{l}_{\frac{N_\nu}{2}-1}\}$ and $\bm{\alpha}^\text{r} = \{\alpha^\text{r}_0,\alpha^\text{r}_1,\ldots,\alpha^\text{r}_{\frac{N_\nu}{2}-1}\}$ can be computed as
\begin{align}
\alpha^\text{l}_i =& 2\arctanh \left(\tanh\left(\frac{\alpha_i}{2}\right)\tanh\left(\frac{\alpha_{i+\frac{N_\nu}{2}}}{2}\right)\right) \text{,} \label{eq1} \\
\alpha^{\text{r}}_i =& \alpha_{i+\frac{N_\nu}{2}} + (1-2\beta^\text{l}_i)\alpha_i \text{,}
\label{eq2}
\end{align}
whereas the $N_\nu$ values of $\bm{\beta}$ are calculated by means of the left and right messages $\bm{\beta}^\text{l} = \{\beta^\text{l}_0,\beta^\text{l}_1,\ldots,\beta^\text{l}_{\frac{N_\nu}{2}-1}\}$ and $\bm{\beta}^\text{r} = \{\beta^\text{r}_0,\beta^\text{r}_1,\ldots,\beta^\text{r}_{\frac{N_\nu}{2}-1}\}$ as
\begin{equation}
\beta_i =
  \begin{cases}
    \beta^\text{l}_i\oplus \beta^\text{r}_i \text{,} & \text{if} \quad i < \frac{N_\nu}{2} \text{,}\\
    \beta^\text{r}_{i-\frac{N_\nu}{2}} \text{,} & \text{otherwise} \text{,}
  \end{cases}
  \label{eq3}
\end{equation}
where $\oplus$ is the bitwise XOR operation. Bits in the left and right child nodes are distinguished by $i<\frac{N_\nu}{2}$. At leaf nodes, the $i$-th bit $\hat{u}_i$ can be estimated as
\begin{equation}
\hat{u}_i =
  \begin{cases}
    0 \text{,} & \text{if } i \in \mathcal{F} \text{ or } \alpha_{i}\geq 0\text{,}\\
    1 \text{,} & \text{otherwise.}
  \end{cases} \label{eq6}
\end{equation}
Equation~(\ref{eq1}) can be reformulated in a more hardware-friendly (HWF) version that has first been proposed in \cite{leroux}:
\begin{equation}
\alpha^{\text{l}}_i = \sgn(\alpha_i)\sgn(\alpha_{i+\frac{N_\nu}{2}})\min(|\alpha_i|,|\alpha_{i+\frac{N_\nu}{2}}|) \text{.} \label{eq4}
\end{equation}

\subsection{Successive-Cancellation List Decoding} \label{sec:polar:SCLDec}

The error-correction performance of SC when applied to codes with short to moderate length can be improved by the use of list-based decoding. The SCL algorithm estimates a bit considering both its possible values $0$ and $1$. At every estimation, the number of codeword candidates (paths) doubles: in order to limit the increase in the complexity of this algorithm, only a set of $L$ codeword candidates is memorized at all times. Thus, after every estimation, half of the paths are discarded. To this purpose, a path metric ($\PM$) is associated to each path and updated at every new estimation: it can be considered a cost function, and the $L$ paths with the lowest $\PM$s are allowed to survive. In the LLR-based SCL \cite{balatsoukas}, the $\PM$ can be computed as
\begin{equation} 
\PM_{i_l} = \sum_{j = 0}^i \ln\left(1+\mathrm{e}^{-(1-2\hat{u}_{j_l})\alpha_{j_l}}\right) \text{,} \label{eq5}
\end{equation}
where $l$ is the path index and $\hat{u}_{j_l}$ is the estimate of bit $j$ at path $l$. A HWF version of Equation~(\ref{eq5}) has been proposed in \cite{balatsoukas}:
\begin{align}
&\PM_{{-1}_l} = 0 \text{,} \nonumber \\
&\PM_{{i}_l} = \begin{cases}
    \PM_{{i-1}_l} + |\alpha_{i_l}| \text{,} & \text{if } \hat{u}_{i_l} \neq \frac{1}{2}\left(1-\sgn\left(\alpha_{i_l}\right)\right)\text{,}\\
    \PM_{{i-1}_l} \text{,} & \text{otherwise,}
  \end{cases} \label{eq7}
\end{align}
which can be rewritten as
\begin{equation}
\PM_{{i}_l} = \frac{1}{2}\sum_{j = 0}^{i}\sgn(\alpha_{{{j}_l}})\alpha_{{{j}_l}} - (1-2\hat{u}_{j_l})\alpha_{{{j}_l}} \text{.} \label{eq7_1}   
\end{equation}

\subsection{Simplified Successive-Cancellation List Decoding} \label{sec:polar:SSCLDec}

\begin{figure}
  \centering
  \begin{tikzpicture}[scale=2.1, thick]
\newcommand\Triangle[1]{-- ++(0:2*#1) -- ++(120:2*#1) --cycle}

  \draw (0,0) circle [radius=.05];
  
  \draw (-.05,0) -- (.05,0);
  \draw (0,-.05) -- (0,.05);

  \draw (-1.05,-.55) \Triangle{.05};
  \fill (1,-.5) circle [radius=.05];

  \draw (0,-.05) -- (-1,-.45);
  \draw (0,-.05) -- (1,-.45);

  \draw [very thin,gray,dashed] (-2,0) -- (2,0);
  \draw [very thin,gray,dashed] (-2,-.5) -- (2,-.5);

  \node at (-1,-.75) {Rep};
  \node at (1,-.75) {Rate-1};
\end{tikzpicture}
  \caption{SSCL decoding tree for $\mathcal{P}(8,5)$ and $\{u_0,u_1,u_2\}\in \mathcal{F}$.}
  \label{figSSCLDec}
\end{figure}
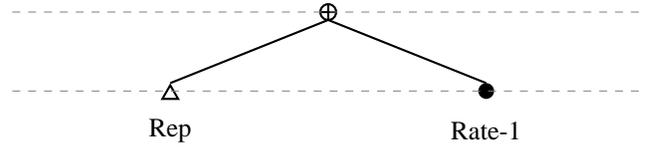

SC decoding requires the traversal of the whole decoding tree. The Fast Simplified SC (Fast-SSC) algorithm in \cite{sarkis} reduces the SC time requirements by exploiting the fact that polar codes are constructed by concatenation of smaller codes. It identifies different constituent codes which can be decoded with efficient maximum likelihood decoding techniques, avoiding traversing parts of the decoding tree. In particular, Fast-SSC makes use of \mbox{Rate-0} nodes which have only frozen bits, \mbox{Rate-1} nodes consisting of information bits only, Repetition (Rep) nodes which have only frozen bits except for the rightmost one, and Single Parity-Check (SPC) nodes that are made of information bits only except for the leftmost one. The advantage of Fast-SSC is that not only it requires fewer time-steps than SC to finish the decoding process, but also it provides an exact match to SC with no error-correction performance loss.

The SSCL algorithm in \cite{hashemi_SSCL} provides efficient decoders for \mbox{Rate-0}, Rep, and \mbox{Rate-1} nodes in SCL without traversing the decoding tree while guaranteeing the error-correction performance preservation. For example in Fig.~\ref{figSCDec}, the black circles represent \mbox{Rate-1} nodes, the white circles represent \mbox{Rate-0} nodes, and the white triangles represent Rep nodes. The pruned decoding tree of SSCL for the example in Fig.~\ref{figSCDec} is shown in Fig.~\ref{figSSCLDec} which consists of a Rep node and a \mbox{Rate-1} node.

Let us consider that the vectors $\bm{\alpha}_l$ and $\bm{\eta}_l = 1-2\bm{\beta}_l$ are relative to the top of a node in the decoding tree. It was shown in \cite{hashemi_SSCL_TCASI} that \mbox{Rate-0} nodes can be decoded in a single time-step as
\begin{subnumcases}{\kern-2em\PM_{{N_\nu-1}_l}\!\!=\!\!}
\!\!\sum_{i = 0}^{N_\nu-1} \ln\left(1+\mathrm{e}^{-\alpha_{i_l}}\right) \text{,} & \kern-1em\text{Exact,} \label{eq:Rate0:Exact} \\
\!\!\frac{1}{2} \sum_{i = 0}^{N-1} \sgn\left(\alpha_{i_l}\right)\alpha_{i_l} - \alpha_{i_l} \text{,} & \kern-1em\text{HWF.} \label{eq:Rate0:HWF}
\end{subnumcases}
Rep nodes can be decoded in two time-steps as
\begin{subnumcases}{\kern-2em\PM_{{N_\nu-1}_l}\!\!=\!\!}
\!\!\sum_{i = 0}^{N_\nu-1} \ln\left(1+\mathrm{e}^{-\eta_{{N_\nu-1}_l}\alpha_{i_l}}\right) \text{,} & \!\!\!\!\kern-1em\text{Exact,} \label{eq:Rep:Exact} \\
\!\!\frac{1}{2} \!\sum_{i = 0}^{N_\nu-1} \!\sgn\left(\!\alpha_{i_l}\!\right)\alpha_{i_l} \!\!-\!\! \eta_{{N_\nu-1}_l}\alpha_{i_l} \text{,} & \!\!\!\!\kern-1em\text{HWF.} \label{eq:Rep:HWF}
\end{subnumcases}
where $\eta_{{N_\nu-1}_l}$ represents the bit estimate of the information bit in the Rep node. Finally, \mbox{Rate-1} nodes can be decoded in $N_\nu$ time-steps as
\begin{subnumcases}{\kern-2em\PM_{{N_\nu-1}_l}\!\!=\!\!}
\!\!\sum_{i = 0}^{N_\nu-1} \ln\left(1+\mathrm{e}^{-\eta_{i_l}\alpha_{i_l}}\right) \text{,} & \!\!\!\!\kern-1em\text{Exact,} \label{eq:Rate1:Exact} \\
\!\!\frac{1}{2} \sum_{i = 0}^{N_\nu-1} \sgn\left(\alpha_{i_l}\right)\alpha_{i_l} - \eta_{i_l}\alpha_{i_l} \text{,} & \!\!\!\!\kern-1em\text{HWF.} \label{eq:Rate1:HWF}
\end{subnumcases}

While the SSCL algorithm reduces the number of required time-steps to decode \mbox{Rate-1} nodes by almost a factor of three, it fails to address the effect of list size on the maximum number of required bit estimations. In the following section, we prove that the number of required time-steps to decode \mbox{Rate-1} nodes depends on the list size and that the new \mbox{Fast-SSCL} algorithm is faster than both SCL and SSCL without incurring any error-correction performance degradation.

\section{Fast Simplified Successive-Cancellation List Decoding}\label{sec:FSSCL}

In this section, we propose a fast decoding approach for \mbox{Rate-1} nodes and use it to develop \mbox{Fast-SSCL}. In order to prove that it is exact and that no approximation is introduced with respect to SCL and SSCL decoding, we first introduce the following lemma.

\begin{lemma} \label{lemma1}
For two positive real numbers $a$ and $b$ where $a<b$, the following holds:
\begin{equation}
\ln \left(1 + \mathrm{e}^{-a}\right) + \ln \left(1 + \mathrm{e}^{b}\right) > \ln \left(1 + \mathrm{e}^{a}\right) + \ln \left(1 + \mathrm{e}^{-b}\right) \label{eq:lemmaEquation}
\end{equation}
\end{lemma}

\begin{proof}
We prove
\begin{equation*}
\ln \left(1 + \mathrm{e}^{-a}\right) + \ln \left(1 + \mathrm{e}^{b}\right) - \ln \left(1 + \mathrm{e}^{a}\right) - \ln \left(1 + \mathrm{e}^{-b}\right) > 0 \text{.}
\end{equation*}
We can write
\begin{align}
& \ln \left(1 + \mathrm{e}^{-a}\right) + \ln \left(1 + \mathrm{e}^{b}\right) - \ln \left(1 + \mathrm{e}^{a}\right) - \ln \left(1 + \mathrm{e}^{-b}\right) = \nonumber \\
& \ln \left(\frac{1 + \mathrm{e}^{-a}}{1 + \mathrm{e}^{a}}\right) + \ln \left(\frac{1 + \mathrm{e}^{b}}{1 + \mathrm{e}^{-b}}\right) = \nonumber \\
& \ln \left(\mathrm{e}^{-a}\frac{1 + \mathrm{e}^{a}}{1 + \mathrm{e}^{a}}\right) + \ln \left(\mathrm{e}^{b}\frac{1 + \mathrm{e}^{-b}}{1 + \mathrm{e}^{-b}}\right) = \nonumber \\
& \ln \left(\mathrm{e}^{-a}\right) + \ln \left(\mathrm{e}^{b}\right) = b-a > 0 \text{,}
\end{align}
which proves the lemma.
\end{proof}

The fast \mbox{Rate-1} decoder can be summarized by the following theorem and its subsequent proof.

\begin{theorem} \label{th:maxEstimate}
In SCL decoding with list size $L$, the maximum number of bit estimations in a \mbox{Rate-1} node of length $N_\nu$ required to get the exact same results as the conventional SCL decoder is
\begin{equation}
\min \left(L-1,N_\nu\right) \text{.}
\end{equation}
\end{theorem}

The proposed technique improves the required number of time-steps to decode \mbox{Rate-1} nodes when $L-1<N_\nu$. Every bit after the $L-1$-th can be obtained through hard decision on the LLR as
\begin{equation}
\beta_{i_l} =
  \begin{cases}
    0 \text{,} & \text{if } \alpha_{i_l}\geq 0\text{,}\\
    1 \text{,} & \text{otherwise,}
  \end{cases} \label{eq:hardDecLLR}
\end{equation}
without the need for path splitting. On the other hand, in case $\min \left(L-1,N_\nu\right)=N_\nu$, all bits of the node need to be estimated and the decoding automatically reverts to the process described in \cite{hashemi_SSCL}. The following proof is nevertheless valid for both $L-1<N_\nu$ and $L-1\geq N_\nu$. 

\begin{proof}

Let us consider the path metrics associated with the $L$ surviving paths at bit estimation step $i$ as $\PM_i = \{\PM_{i_0},\ldots,\PM_{i_{L-1}}\}$ and the LLR values associated with the \mbox{Rate-1} node at path $l$ as $\bm{\alpha}_l = \{\alpha_{0_l},\alpha_{1_l},\ldots,\alpha_{{N_\nu-1}_l}\}$. For the purpose of this proof, let us also consider the vectors $\PM_i$ and $\bm{\alpha}_l$ sorted as follows:
\begin{align*}
 &\PM_{i_l} \leq \PM_{i_{l+1}}, &&0 \leq l < L-1\text{,} \\
 &|\alpha_{i_l}| \leq |\alpha_{{i+1}_l}|, &&0 \leq i < N_\nu-1\text{.}
\end{align*}

At each estimation step $i$, the corresponding bit is estimated as either $0$ or $1$, and the $\PM$s are updated as
\begin{equation}
\PM_{i_l} = \sum_{j = 0}^i \ln\left(1+\mathrm{e}^{-\eta_{j_l}\alpha_{j_l}}\right) \text{,} \label{eq:PMTheorem}
\end{equation}
which is a monotonic and non-decreasing function of $i$. At any given estimation step $i$ within the \mbox{Rate-1} node, the least reliable LLR among those still to be estimated is $\alpha_{i_l}$.

We now prove the theorem by contradiction. Let us suppose that step $L-1$ splits path $l$ into two surviving paths. The corresponding $\PM$s will be
\begin{align}
\PM_{{L-1}_p}\! &= \!\sum_{j = 0}^{L-2} \ln\left(1+\mathrm{e}^{-\eta_{j_l}\alpha_{j_l}}\right) \!+\! \ln\left(1+\mathrm{e}^{-|\alpha_{{L-1}_l}|}\right) \text{,} \\
\PM_{{L-1}_q}\! &= \!\sum_{j = 0}^{L-2} \ln\left(1+\mathrm{e}^{-\eta_{j_l}\alpha_{j_l}}\right) \!+\! \ln\left(1+\mathrm{e}^{|\alpha_{{L-1}_l}|}\right) \text{,}
\end{align}
where $0 \leq p < q < L$. We now show that there are at least $L$ bit estimation sequences that result in $\PM$s which are less than $\PM_{{L-1}_q}$. To this end, we demonstrate that there are $L$ paths originated from path $l$ with smaller $\PM$s than $\PM_{{L-1}_q}$ that are generated before estimating bit $L-1$.

Let us consider the lowest possible value that $\PM_{{L-1}_q}$ can assume:
\begin{equation}
\PM_{{L-1}_q} = \sum_{j = 0}^{L-2} \ln\left(1+\mathrm{e}^{-|\alpha_{j_l}|}\right) + \ln\left(1+\mathrm{e}^{|\alpha_{{L-1}_l}|}\right) \text{,} \label{eq:PMWorstCase}
\end{equation}
which represents the case where the bits estimated in steps $0\le j \le L-2$ match the hard decision of their corresponding LLR values, and the $L-1$-th does not. Let us now consider the bit sequences differing from path $q$ in that the bit that does not match the LLR hard decision is at step $w$, where $0\le w \le L-2$, while the $L-1$-th matches. The corresponding $\PM$ would be
\begin{align}
\PM_{{L-1}_v} =& \sum_{\substack{j = 0\\j \neq w}}^{L-2} \ln\left(1+\mathrm{e}^{-|\alpha_{j_l}|}\right) \nonumber\\
&+ \ln\left(1+\mathrm{e}^{|\alpha_{{w}_l}|}\right) + \ln\left(1+\mathrm{e}^{-|\alpha_{{L-1}_l}|}\right) \text{.} \label{eq:PMestimate}
\end{align}
Rewriting Equation (\ref{eq:PMWorstCase}) as
\begin{align}
\PM_{{L-1}_q} =& \sum_{\substack{j = 0\\j \neq w}}^{L-2} \ln\left(1+\mathrm{e}^{-|\alpha_{j_l}|}\right) \nonumber\\
&+ \ln\left(1+\mathrm{e}^{-|\alpha_{{w}_l}|}\right) + \ln\left(1+\mathrm{e}^{|\alpha_{{L-1}_l}|}\right) \text{,}
\end{align}
and using the fact that $|\alpha_{{L-1}_l}| > |\alpha_{{w}_l}|$, we can use the result in Lemma \ref{lemma1} to conclude
\begin{equation}
\PM_{{L-1}_q} > \PM_{{L-1}_v} \text{,}
\end{equation}
which in turn results in $q > v$. Since $w$ can assume $L-1$ values, and taking in account the bit sequence represented by path $p$ where all the bits agree with their corresponding LLR hard decision, there are at least $L$ bit sequences which result in a smaller $\PM$ than $\PM_{{L-1}_q}$. Therefore, $q \geq L$ which contradicts the assumption that $q < L$ and confirms that path $q$ will be discarded. In other words, this proves that paths that consider bits not matching the LLR hard decision after the $L-1$-th step will always be discarded: it is thus useless to split paths after the $L-1$-th. Theorem~\ref{th:maxEstimate} is consequently proven.

\end{proof}

The proposed theorem remains valid also for the HWF formulation that can be written as
\begin{equation}
\PM_{{i}_l} = \begin{cases}
    \PM_{{i-1}_l} + |\alpha_{i_l}|, & \text{if } \eta_{i_l} \neq \sgn\left(\alpha_{i_l}\right)\text{,}\\
    \PM_{{i-1}_l}, & \text{otherwise,}
  \end{cases} \label{PMupdate0}
\end{equation}
At each step $i$, depending on the value of $|\alpha_{i_l}|$, two cases arise.
\begin{itemize} 
 \item [{\bf A}]$|\alpha_{i_l}| \geq \PM_{{i-1}_{L-1}} - \PM_{{i-1}_l}$ 
 
From (\ref{PMupdate0}), we can see that the modulus of the least reliable bit $|\alpha_{i_l}|$ is the minimum quantity that can be added to the $\PM$ in case $\eta_{i_l} \neq \sgn\left(\alpha_{i_l}\right)$. If this quantity is greater than the difference between the currently considered path metric $\PM_{{i-1}_l}$ and the largest surviving path metric $\PM_{{i-1}_{L-1}}$, every estimation that sees $\eta_{i_l} \neq \sgn\left(\alpha_{i_l}\right)$ will lead to $\PM_{{i-1}_l}+|\alpha_{i_l}|\geq \PM_{{i-1}_{L-1}}$ and thus to a discarded path. Consequently, for all the remaining estimations in the \mbox{Rate-1} node, paths need not to be duplicated, and bits are estimated as $\eta_{i_l} = \sgn\left(\alpha_{i_l}\right)$.
 
 \item [{\bf B}]$|\alpha_{i_l}| < \PM_{{i-1}_{L-1}} - \PM_{{i-1}_l}$
 
Let us consider positions $p$ and $q$ in the ordering of $\PM$s such that
\begin{align*}
\PM_{{i}_{p}} &= \PM_{{i-1}_l} \text{,} \\
\PM_{{i}_{q}} &= \PM_{{i-1}_l} + |\alpha_{i_l}| \text{,}
\end{align*}
where $l \leq p < q < L$. In this case, both bit estimates for the least reliable bit have to be taken into account since their corresponding paths will be ordered among the first $L$. In turn, the path in position $L-1$ at step $i-1$ is moved to position $L$ at step $i$ and thus discarded. The following estimation step $i+1$ must be evaluated independently, to see if it falls in case {\bf A} or {\bf B}.
\end{itemize}
As soon as case {\bf A} is encountered in path $l$, that path does not need to undergo any subsequent splitting, and the remaining $\beta_{i_l}$ can be obtained through LLR hard decision of Equation (\ref{eq:hardDecLLR}). While case {\bf B} requires continued path splitting, this can occur a limited amount of times before case {\bf A} is encountered. The maximum amount of consecutive case {\bf B} occurrences can be identified by the following worst case analysis.
\begin{enumerate}
 \item Case {\bf B} occurs at $i=0$ and $l=0$.
 \item Considering that $\PM_{-1_{0}}$ is the $\PM$ at $l=0$ before the first bit of the \mbox{Rate-1} node is estimated, if $p=0$ and $q=1$ then
\begin{align*}
 &\PM_{0_{0}}=\PM_{-1_{0}}\\
 &\PM_{0_{1}}=\PM_{-1_{0}}+|\alpha_{0_0}|\\
 &\PM_{0_{2}}=\PM_{-1_{1}}\\
 & \vdots
\end{align*}
 \item Case {\bf B} occurs at $i=1$ and $l=0$.
 \item Since $|\alpha_{0_0}|\leq|\alpha_{1_0}|$, $q > 1$. For $p=0$ and $q=2$,
\begin{align*}
 &\PM_{1_{0}}=\PM_{-1_{0}}\\
 &\PM_{1_{1}}=\PM_{-1_{0}}+|\alpha_{0_0}|\\
 &\PM_{1_{2}}=\PM_{-1_{0}}+|\alpha_{1_0}|\\
 & \vdots
\end{align*}
 \item Since at every step $|\alpha_{i-1_0}|\leq|\alpha_{i_0}|$, then $q > i$. If at every case {\bf B} step $p = 0$ and $q=i+1$, a total of $L-1$ consecutive case {\bf B} are possible, after which $q>L-1$, resulting in case {\bf A}. At $i=L-2$, the $L$ surviving $\PM$s after $L-1$ consecutive case {\bf B} are the following:
 \begin{align*}
 &\PM_{L-2_{0}}=\PM_{-1_{0}}\\
 &\PM_{L-2_{1}}=\PM_{-1_{0}}+|\alpha_{0_0}|\\
 &\PM_{L-2_{2}}=\PM_{-1_{0}}+|\alpha_{1_0}|\\
 & \vdots \\
 &\PM_{L-2_{L-1}}=\PM_{-1_{0}}+|\alpha_{L-2_0}|\text{.}
\end{align*}
\end{enumerate}
Much like the case considered in the proof for Theorem~\ref{th:maxEstimate}, the above analysis shows that at most $L-1$ bit estimations are required to guarantee the exact same results as the conventional SCL. Thus, the theorem is valid also with the HWF Equation~(\ref{PMupdate0}).

In the presented proof and discussion, $\bm{\alpha}_l$ and $\PM_i$ are assumed to be sorted at every step for the sake of simplicity. $\PM$s are sorted every time paths are split, i.e. when an information bit is estimated. When the decoding process considers frozen bits, paths are not split and even if modified, $\PM$s retain their ordering. On the other hand, $\bm{\alpha}_l$ is not ordered, but at each step $i$ the full vector sorting can be substituted with the identification of the minimum $|\alpha_{i_l}|$.

The result of Theorem~\ref{th:maxEstimate} provides an exact number of bit-estimations in \mbox{Rate-1} nodes for each list size in SCL decoding in order to guarantee error-correction performance preservation. The \mbox{Rate-1} node decoder of \cite{sarkis_list} states that two bit-estimations are required to preserve the error-correction performance, but this result is found empirically. The following remarks are the direct results of Theorem~\ref{th:maxEstimate}.

\begin{remark}
The \mbox{Rate-1} node decoder of \cite{sarkis_list} for $L=2$ is redundant.
\end{remark}
Theorem~\ref{th:maxEstimate} states that for a \mbox{Rate-1} node of length $N_\nu$ when $L=2$, the number of bit-estimations is $\min(L-1,N_\nu) = 1$. Therefore, there is no need to estimate the bits after the least reliable bit is estimated. \cite{sarkis_list} for $L=2$ is thus redundant.

\begin{remark}
The \mbox{Rate-1} node decoder of \cite{sarkis_list} falls short in preserving the error-correction performance for higher rates and larger list sizes.
\end{remark}
For codes of higher rates, the number of \mbox{Rate-1} nodes of larger length increases \cite{hashemi_SSCL_TCASI}. Therefore, when the list size is also large, $\min(L-1,N_\nu) \gg 2$. The gap between the empirical method of \cite{sarkis_list} and the result of Theorem~\ref{th:maxEstimate} can introduce significant error-correction performance loss. Fig.~\ref{fig:ER1kL128_1024_860} provides the frame error rate (FER) and bit error rate (BER) curves obtained with \mbox{Fast-SSCL} decoding ($L=128$) for a $\mathcal{P}(1024,860)$ code. The code is concatenated with a cyclic redundancy check of length $32$, and different curves are provided for the \mbox{Rate-1} node decoder in Theorem~\ref{th:maxEstimate}, and the empirical method of \cite{sarkis_list}. It can be seen that the error-correction performance loss reaches $0.25$dB at FER of $10^{-5}$.

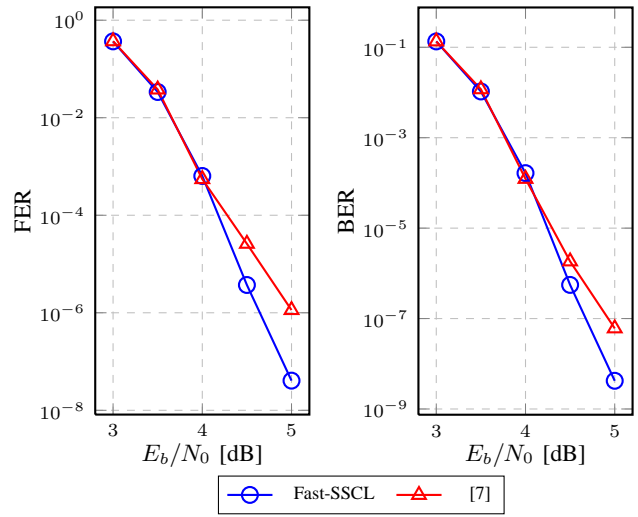
\begin{figure}
  \centering
  \hspace{-15pt}
  \begin{tikzpicture}
  \pgfplotsset{
    label style = {font=\fontsize{9pt}{7.2}\selectfont},
    tick label style = {font=\fontsize{7pt}{7.2}\selectfont}
  }

\begin{axis}[
	scale = 1,
    ymode=log,
    xlabel={$E_b/N_0$ [\text{dB}]}, xlabel style={yshift=0.8em},
    ylabel={FER}, ylabel style={yshift=-0.75em},
    grid=both,
    ymajorgrids=true,
    xmajorgrids=true,
    grid style=dashed,
    width=0.5\columnwidth, height=7cm,
    thick,
    mark size=3,
    legend style={
      anchor={center},
      cells={anchor=west},
      column sep= 2mm,
      font=\fontsize{7pt}{7.2}\selectfont,
    },
    legend to name=perf-legend1kL128_1024_860,
    legend columns=2,
]

\addplot[
    color=blue,
    mark=o,
    thick,
    mark size=3,
]
table {
3 0.3681
3.5 0.0337
4 0.000639125
4.5 3.71452e-06
5 4.05754e-08
};
\addlegendentry{Fast-SSCL}

\addplot[
    color=red,
    mark=triangle,
    thick,
    mark size=3,
]
table {
3 0.37
3.5 0.0379
4 0.000550594
4.5 2.60017e-05
5 1.13644e-06
};
\addlegendentry{\cite{sarkis_list}}

\end{axis}
\end{tikzpicture}
  \begin{tikzpicture}
  \pgfplotsset{
    label style = {font=\fontsize{9pt}{7.2}\selectfont},
    tick label style = {font=\fontsize{7pt}{7.2}\selectfont},
  }

\begin{axis}[
	scale = 1,
    ymode=log,
    xlabel={$E_b/N_0$ [\text{dB}]}, xlabel style={yshift=0.8em},
    ylabel={BER}, ylabel style={yshift=-0.75em},%
    grid=both,
    ymajorgrids=true,
    xmajorgrids=true,
    width=0.5\columnwidth, height=7.0cm,
    grid style=dashed,
    thick,
    mark size=3,
]

\addplot[
    color=blue,
    mark=o,
    thick,
    mark size=3,
]
table {
3 0.134853
3.5 0.0105408
4 0.000164701
4.5 5.55019e-07
5 4.20463e-09
};

\addplot[
    color=red,
    mark=triangle,
    thick,
    mark size=3,
]
table {
3 0.134756
3.5 0.0117945
4 0.000122962
4.5 1.83733e-06
5 6.06804e-08
};

\end{axis}
\end{tikzpicture}
  \\
  \hspace{20pt}\ref{perf-legend1kL128_1024_860}\vspace{2pt}
  \caption{FER and BER performance comparison of decoding $\mathcal{P}(1024,860)$ with \mbox{Fast-SSCL} and the empirical method of \cite{sarkis_list} when $L=128$. The cyclic redundancy check length is $32$.}
  \label{fig:ER1kL128_1024_860}
\end{figure}

The proposed \mbox{Rate-1} node decoder is used in the \mbox{Fast-SSCL} algorithm, while the decoders for \mbox{Rate-0} and Rep nodes remain similar to those used in SSCL \cite{hashemi_SSCL}. In the following section, we show that in a polar code, there are many instances where $L-1<N_\nu$ for \mbox{Rate-1} nodes and using the \mbox{Fast-SSCL} algorithm can significantly reduce the number of required decoding time-steps with respect to both SCL and SSCL.

\section{Time-Step Reduction} \label{sec:impr}

In Section~\ref{sec:FSSCL}, we have demonstrated that when $L-1<N_\nu$, up to $L-1$ bit estimations are necessary when decoding \mbox{Rate-1} nodes to match the performance of SCL and SSCL. The time-step reduction for the complete polar code decoding that can be gained through this technique, however, depends on the structure of the code itself. As an example, Table \ref{tab:specialNodes} shows the number of time-steps required to decode a polar code with $N=1024$ optimized for $E_b/N_0 = 2$ dB: results are given for three different rates, five list sizes, and SCL, SSCL and \mbox{Fast-SSCL} decoding algorithms.

\begin{table}
	\centering
	\caption{Number of Time-Steps for SCL, SSCL, and \mbox{Fast-SSCL} Decoding of a Polar Code of Length $N=1024$. The Code is Optimized for $E_b/N_0 = 2$ {\rm dB}.}
	\label{tab:specialNodes}
		\setlength{\extrarowheight}{1.7pt}
	\begin{tabularx}{.37\textwidth}{ccYYY}
\toprule

Rate & $L$ & SCL & SSCL & \mbox{Fast-SSCL} \\

\midrule

\multirow{5}{*}{$\displaystyle\frac{1}{4}$} & $2$ & $2302$ & $533$ & $394$ \\
& $4$ & $2302$ & $533$ & $474$ \\
& $8$ & $2302$ & $533$ & $518$ \\
& $16$ & $2302$ & $533$ & $532$ \\
& $32$ & $2302$ & $533$ & $533$ \\

\midrule

\multirow{5}{*}{$\displaystyle\frac{1}{2}$} & $2$ & $2558$ & $793$ & $397$ \\
& $4$ & $2558$ & $793$ & $500$ \\
& $8$ & $2558$ & $793$ & $597$ \\
& $16$ & $2558$ & $793$ & $687$ \\
& $32$ & $2558$ & $793$ & $757$ \\

\midrule

\multirow{5}{*}{$\displaystyle\frac{3}{4}$} & $2$ & $2814$ & $1001$ & $334$ \\
& $4$ & $2814$ & $1001$ & $435$ \\
& $8$ & $2814$ & $1001$ & $545$ \\
& $16$ & $2814$ & $1001$ & $667$ \\
& $32$ & $2814$ & $1001$ & $801$ \\

\bottomrule
	\end{tabularx}
\end{table}

It can be observed that the required number of time-steps for SCL and SSCL does not depend on the list size $L$, but just on the code rate, and thus on the number and size of the constituent codes. Low code rates can exploit a higher number and larger size of \mbox{Rate-0} and Rep nodes, thus the reduction in the number of time-steps required for SSCL over SCL reaches $76.8\%$ at rate $1/4$ against $64.4\%$ at rate $3/4$.

The number of decoding time-steps for the \mbox{Fast-SSCL} algorithm, on the other hand, depends on $L$: a small list size will result in a fast decoding process, that will however yield lower error-correction performance with respect to a larger list size, but will not degrade it with respect to SCL and SSCL with the same $L$. The larger advantages can be observed for high rates, where \mbox{Rate-1} nodes are more numerous and are larger: with $L=2$ and rate $3/4$, \mbox{Fast-SSCL} requires $66.6\%$ and $88.1\%$ fewer time-steps than SSCL and SCL respectively, without causing any deterioration in error-correction performance.

\section{Conclusion} \label{sec:conc}

In this work, we have proposed a faster approach to the decoding of \mbox{Rate-1} nodes in polar codes which resulted in the development of \mbox{Fast-SSCL} decoding algorithm. We have postulated and demonstrated that the number of bit estimations and consequent path splitting of a \mbox{Rate-1} node of length $N_\nu$ necessary to exactly match the error-correction performance of SCL or SSCL decoding with list size $L$ is $\min \left(L-1,N_\nu\right)$. Considering a set of codes as a case study, we have shown that the whole polar code decoding process can benefit in time-step reduction of up to $88.1\%$ with respect to SCL, and $66.6\%$ with respect to SSCL decoding algorithm without any kind of error-correction performance degradation.

\end{document}